\documentclass[11pt]{article} 
\usepackage[margin=1.6in]{geometry}
\usepackage[utf8]{inputenc}
\usepackage[T1]{fontenc}
\usepackage[english]{babel}
\usepackage{lmodern} 
\usepackage{graphicx}
\usepackage{wrapfig}
\usepackage{subfig}
\usepackage{caption}
\usepackage{paralist}

\graphicspath{{./figures/}}

\newcommand{\titel}{Computing Vertex-Disjoint Paths using MAOs}

\usepackage{algorithm} 
\usepackage[noend]{algpseudocode} 

\usepackage[usenames]{color} 
\definecolor{hellblau}{rgb}{0.2,0.4,1} 
\definecolor{dunkelblau}{rgb}{0,0,0.8}
\definecolor{dunkelgruen}{rgb}{0,0.5,0}
\usepackage[
	pdftex,
	colorlinks,
	linkcolor=dunkelblau,
	urlcolor=dunkelblau,
	citecolor=dunkelgruen,
	bookmarks=true,
	linktocpage=true,
	pdftitle={\titel},
	pdfsubject={}
]{hyperref} 
\usepackage{pdfpages}

\usepackage{xspace}
\usepackage{pdfpages}

\usepackage{amsmath} 
\usepackage{amsthm} 
\usepackage{amsfonts} 
\theoremstyle{plain} 
	\newtheorem{satz}{Satz}[] 
	\newtheorem{theorem}[satz]{Theorem}
	\newtheorem{lemma}[satz]{Lemma}
	\newtheorem{observation}[satz]{Observation}
	
	\newtheorem*{invariants}{Invariants}
\theoremstyle{remark} 
	\newtheorem*{remark}{\textbf{Remark}} 
	\newtheorem*{problem}{Open Problem}
\theoremstyle{definition} 

	\newtheorem{corollary}[satz]{Corollary}

\newcommand{\letzter}{\emph{end}}
\newcommand{\vorletzter}{\emph{sec}}
\newcommand{\links}{\emph{left}}

\begin{document}
	\title{\titel}
		\author{Johanna E.\ Preißer\footnote{This research was supported by the DFG grant SCHM 3186/1-1.}\\
		Institute of Mathematics\\
		TU Ilmenau \and
		Jens M. Schmidt\footnotemark[1]\\
		Institute of Mathematics\\
		TU Ilmenau}
	\date{}
	\maketitle

\begin{abstract}
Let $G$ be a graph with minimum degree $\delta$. It is well-known that maximal adjacency orderings (MAOs) compute a vertex set $S$ such that every pair of $S$ is connected by at least $\delta$ internally vertex-disjoint paths in $G$.

We present an algorithm that, given any pair of $S$, computes these $\delta$ paths in linear time $O(n+m)$. This improves the previously best solutions for these special vertex pairs, which were flow-based. Our algorithm simplifies a proof about pendant pairs of Mader and makes a purely existential proof of Nagamochi algorithmic.
\end{abstract}

\section{Introduction}
Vertex-connectivity is a fundamental parameter of graphs that, according to Menger~\cite{Menger1927}, can be characterized by the existence of internally disjoint paths between vertex pairs. Thus, much work has been devoted to the following question: Given a number $k$, a simple graph $G=(V,E)$, and two vertices of $G$, compute $k$ internally disjoint paths between these vertices if such paths exist. Despite all further efforts, the traditional flow-based approach by Even and Tarjan~\cite{Even1975} and Karzanov~\cite{Karzanov1973} gives still the best bound $O(\min\{k,\sqrt{n}\}m)$ for this task, where $n := |V|$ and $m := |E|$.

Our research is driven by the question whether $k$ internally disjoint paths can be computed faster. We have no general answer, but show that for specific pairs of vertices, this can actually be done using \emph{maximal adjacency orderings (MAOs)}. MAOs order the vertices of a graph, are also known under the name \emph{maximum cardinality search}, and can be computed in time $O(n+m)$~\cite{Tarjan1984}; we give a definition of MAOs later. On a side note, we exhibit apparently forgotten similarities of MAOs to a very early proof by Mader about pendant pairs in 1971~\cite{Mader1973,Mader1971b}. In fact, MAOs (and maximum cardinality search) can be seen as a simpler version of Mader's result.

Consider a simple graph $G$ with minimum degree $\delta$ and a MAO $<$ of $G$. Let a subset $S$ of vertices be \emph{$k$-connected} if $G$ contains $k$ internally vertex-disjoint paths between every two vertices of $S$. Henzinger proved for every $1 \leq k \leq \delta$ that the last $\delta-k+2$ vertices of $<$ are $k$-connected~\cite{Henzinger1997a}. Nagamochi~\cite{Nagamochi2006} improved this result to the more general vertex sets of trees derived from the \emph{forest decomposition} of the MAO; one of these trees actually contains the last $\delta-k+2$ vertices given above. Nagamochi's proof uses the machinery of mixed connectivity, which generalizes both edge- and vertex-connectivity.

If we desire $\delta$-connected vertex sets, this may result only in a set of size two in the worst case. However, in practice, many more of these sets may occur and each of them may have a much larger size. Even in the worst-case, the approach allows us to increase the size $\delta-k+2$ of these sets at the expense of linearly decreasing their connectivity to any prescribed $k < \delta$.

We consider the problem of computing $k$ internally vertex-disjoint paths between any vertex pair $\{s,t\}$ that is contained in one of the above $k$-connected sets. By the above results, we already know that these paths exist. However, the proofs of Henzinger and Nagamochi are non-constructive and therefore do not give any faster algorithm for finding these paths than the traditional flow-based one.

We give an algorithm that computes these $k$ paths in linear time $O(n+m)$. The algorithm proceeds by a right-to-left sweep in the MAO through the vertices of the corresponding set. We also show how the computation can be extended to find the $k$ internally vertex-disjoint paths between a vertex and a vertex set, and between two vertex sets, whose existence was shown by Menger~\cite{Menger1927}.

\paragraph{Certifying Algorithms.} Being able to compute $k$ internally vertex-disjoint paths has a benefit that purely existential proofs or algorithms that only argue about vertex separators do not have: It certifies the connectivity between the two vertices. This has been used for \emph{certifying algorithms}, as proposed in~\cite{McConnell2011}.

E.g., the famous \emph{min-cut} algorithm of Nagamochi and Ibaraki~\cite{Nagamochi1992b,Frank1994}, as simplified by Stoer and Wagner~\cite{Stoer1997} (and initiated by Mader) uses the $k$-connected vertex sets above by just setting $k := \delta$ and contracting the last two $\delta$-connected (and thus $\delta$-edge-connected) vertices of a MAO iteratively. For unweighted multigraphs, this can be easily made certifying: The certificate just stores the $k$ edge-disjoint paths between the last two vertices $\{s,t\}$ of the MAO in every step; the global $k$-edge-connectivity then follows by transitivity. In fact, the desired $k$ edge-disjoint paths are given directly by considering the (edge-disjoint) trees $T_1,\ldots,T_k$ of the MAO that contain $t$ and taking the unique $s$-$t$-path in each of them.

Using more involved methods, Arikati and Mehlhorn~\cite{Mehlhorn1999b} showed that the Nagamochi-Ibaraki algorithm can be made certifying even for weighted graphs, again without increasing the asymptotic running time and space. For the problem of testing a graph on being $k$-(vertex-)connected, finding a good certificate seems to be an open graph theoretic problem, even when $k$ is fixed:

\begin{problem}
For $k \in \mathbb{N}$ and a simple graph $G$, find a small and easy-to-verify certificate that a graph is $k$-connected.
\end{problem}

So far, linear-time certifying algorithms are only known for $k \leq 3$~\cite{Whitney1932a,Schmidt2013}. For $k \geq 4$, the best known certifying algorithm still seems to be the traditional flow-based one~\cite{Even1975}, which achieves a running time of $O((k+\sqrt{n})k\sqrt{n}m)$. Should at some day the best algorithm for computing vertex-connectivity be based on MAOs, our data structure provides the means to turn this into an efficient certificate.

\section{Maximal Adjacency Orderings}\label{sec:prel}
Throughout this paper, our input graph $G=(V,E)$ is simple, unweighted and of minimum degree $\delta$. We assume standard graph theoretic notation as in~\cite{Diestel2010}. A \emph{maximal adjacency ordering} (MA Ordering; MAO) $<$ of $G$ is a total order $1,\dots,n$ on $V$ such that, for every two vertices $v < w$, $v$ has at least as many neighbors in $\{1, \dots, v-1\}$ as $w$ has. For ease of notation, we always identify the vertices of $G$ with their position in $<$.

Every MAO $<$ partitions $E$ into forests $F_1,\dots,F_m$ (some of which may be empty) as follows: If $v>1$ is a vertex of $G$ and $w_1 < \dots < w_l$ are the neighbors of $v$ in $\{1,\dots,v-1\}$, the edge $\{w_i,v\}$ belongs to $F_i$ for all $i \in \{1,\dots,l\}$. For every $i$, the graph $(V,F_i)$ is a forest (we refer to~\cite[Section~2.2]{Nagamochi2008} for a proof). For the sake of conciseness, we identify this forest with its edge set $F_i$. The partition of $E$ into these forests is called the \emph{forest decomposition} of $<$. For vertices $v < w$, we say $v$ is \emph{left} of $w$. If there is an edge between $v$ and $w$, we call this a \emph{left-edge} of $w$.

For any $k$, we allow to compute $k$ internally vertex-disjoint paths between any two vertices that are contained in a tree $T_k$ of the forest $F_k$. Hence, throughout the paper, let $s>1$ be an arbitrary but fixed vertex of $G$ and let $k$ be a positive integer that is at most the number of left-edges of $s$. The vertex $s$ will be the start vertex of the $k$ internally vertex-disjoint paths to find (the end vertex will be left of $s$). E.g., if we choose $s$ as the last vertex of the MAO (or any other vertex with at least that many left-edges), $k$ can be chosen as any value that is at most the degree of vertex $n$; in particular, $k$ can be chosen arbitrary in the range $1,\dots,\delta$, as claimed in the introduction.

For $i \in \{1, \dots, k\}$, let $T_i$ the component of $F_i$ that contains $s$. As $i \leq k$, $T_i$ is a tree on at least two vertices. Let the smallest vertex $r_i$ of $T_i$ with respect to $<$ be the \emph{root} of $T_i$. For the purpose of this paper, it suffices to consider the subgraph of $G$ induced by the edges of $T_1, \dots, T_k$.

\begin{lemma}[{\cite[Lemma 2.25]{Nagamochi2008}}]\label{lem:interval}
Let $i\in \{1, \dots, k\}$. Then $V(T_i)$ consists of the consecutive vertices $r_i, r_i+1, \dots, w$ in $<$ such that $s \leq w$. Moreover, for each vertex $v \in T_i \setminus \{r_i\}$, the vertex set $\{r_i, r_i+1, \dots, v\}$ induces a connected subgraph of $T_i$.
\end{lemma}

Hence, for every $i\in \{1, \dots,k\}$, every vertex $v>r_i$ of $T_i$ has exactly one left-edge that is in $T_i$ and thus at least $i$ left-edges that are in $G$. Let $\links_i(v)$ be the end vertex of the left-edge of $v$ in $F_i$. The root $r_i$ of $T_i$ has left-degree exactly $i-1$, as if it had more, $r_i$ would have a left-edge in $F_i$ and thus not be the root of $T_i$ and, if it had less, the left-degree of $r_i+1$ cannot be at least $i$, as this violates the MA Ordering (this uses that $G$ is simple).

We conclude that $r_1 < r_2 \dots <r_k$. This, the definition of $F_i$ and Lemma~\ref{lem:interval} imply the following corollary.

\begin{corollary}\label{cor:leftedges}
Let $i < j \leq k$ and let $v$ be a vertex with $r_j < v < s$. Then $v$ is in $T_j$ and $T_i$, $r_i\leq \links_i(v) < \links_j(v) <v$ and $r_j \leq \links_j(v)$.
\end{corollary}

For a vertex-subset $S \subseteq V$, let $\overline{S} := V \setminus S$. For convenience, we will denote sets $\overline{\{v\}}$ by $\overline{v}$. For a vertex-subset $S \subseteq V$, a set of paths is $S$-\emph{disjoint} if no two of them intersect in a vertex that is contained in $S$. Thus, $V$-disjointness is the usual vertex-disjointness and a set of paths is $\overline{v}$-\emph{disjoint} if every two of them intersect in either the vertex $v$ or not at all.

We represent paths as lists of vertices. The \emph{length} of a path is the number of edges it contains. For a path $A$, let $\letzter(A)$ be the last vertex of this list and, if the path has length at least one, let $\vorletzter(A)$ be the second to last vertex of this list.

\paragraph{A side note on a proof by Mader.} In \cite{Mader1973,Mader1971b}, Mader presents a method to find a \emph{pendant pair} $\{v,w\}$ in a simple graph, which is a pair of vertices that is $\min\{d(v),d(w)\}$-connected, where $d$ is the degree function. 
He chooses an inclusion maximal clique and deletes vertices of this clique until it is not maximal anymore. When deleting vertices, new edges are added to preserve the degree of vertices that are not in the clique. Next the clique is enlarged to a maximal one; this procedure is iterated until every edge has an incident vertex in the clique. Then a pendent pair of the original graph can be found. It turns out that this method is a preliminary version of MAOs: The order in which the vertices are added to the clique is in fact a maximal adjacency ordering, and, for every $k$-connected pendent pair found by this method, there is a MAO whose forest decomposition contains this pair as the end vertices of an edge in $F_k$.  This seems to be widely unknown and we are only aware of one place in literature where this similarity is (briefly) mentioned~\cite[p.~443]{Mader1996}. Mader's existential proof can in fact be made algorithmic. However, MAOs provide a much nicer structure, as they work directly on the original graph.

\section{The Loose Ends Algorithm}
We first consider the slightly weaker problem of computing $k$ internally vertex-disjoint paths between $s$ and the root set $\{r_1,\ldots,r_k\}$. We will extend this to compute $k$ internally vertex-disjoint paths between two vertices in the next section.

\begin{lemma}\label{lem:looseends}
Algorithm~\ref{alg:looseEnds} computes $k$ $\overline{s}$-disjoint paths in $T_1 \cup \cdots \cup T_k$  from $s$ to $\{r_1,\ldots,r_k\}$ in time $O(|E(T_1 \cup \dots \cup T_k)|) \subseteq O(n+m)$. 
\end{lemma}

The outline of our algorithm is as follows. We initialize each $A_i$ to be the path that consists of the two vertices $s$ and $\links_i(s)$ (in that order). The vertices $\links_i(s)$ are marked as \emph{active}; throughout the algorithm, let a vertex be \emph{active} if it is an end vertex of an unfinished path $A_i$.

So far the $A_i$ are $\overline{s}$-disjoint. We aim for augmenting each $A_i$ to $r_i$. Step by step, for every active vertex $v$ from $s-1$ down to $r_1$ in $<$, we will modify the $A_i$ to longer paths, similar as in sweep line algorithms from computational geometry. The modification done at an active vertex $v$ is called a \emph{processing step}. From a high-level perspective, the end vertices of several paths $A_i$ may be replaced or augmented by new end vertices $w$ such that $r_i \leq w < v$ during the processing step of $v$. Such vertices $w$ are again marked as active, which results in a continuous modification of each $A_i$ to a longer path. By the above restriction on $w$, each path $A_i$ will have strictly decreasing vertices in $<$ throughout the algorithm. At the end of the processing step of $v$, we unmark $v$ from being active.

Let $v$ be the active vertex that is largest in $<$. Assume that $v$ is the end vertex of exactly one $A_i$. If $v = r_i$, $A_i$ is finished. Otherwise, we append the vertex $\links_i(v)$ to $A_i$ (see Algorithm~\ref{alg:looseEnds}). The important aspect of this approach is that the index of the path $A_i$ predetermines the vertex that augments $A_i$. Clearly, this way $A_i$ will reach $r_i$ at some point, according to Lemma~\ref{lem:interval}.

\begin{algorithm}[htb]
\caption{LooseEnds($G,<,s,k$)}\label{alg:looseEnds}
\begin{algorithmic}[1]
	\ForAll{$i$}\Comment{initialize all $A_i$}
		\State $A_i := (s,\links_i(s))$
		\State Mark $\links_i(s)$ as active
	\EndFor
	\While{there is a largest active vertex $v$}\Comment{process $v$}
		\State Let $j_1 < j_2 < \cdots < j_l$ be the indices of the paths $A_{j_i}$ that end at $v$\label{line:enda}
		\For{$i := 2$ to $l$}\Comment{replace end vertices}
			\State Replace $\letzter(A_{j_i})$ with $\links_{j_{i-1}}(\vorletzter(A_{j_i}))$\label{line:Repl}
			\State Mark $\links_{j_{i-1}}(\vorletzter(A_{j_i}))$ as active
		\EndFor
		\State Perform a cyclic downshift on $A_{j_1},\ldots,A_{j_l}$\Comment{$A_{j_i} := A_{j_{i+1}}$, $A_{j_l} := A_{j_1}$}\label{line:afterRepl}
		\If{$v = r_{j_l}$}
			\State $A_{j_l}$ is finished\Comment{$r_{j_l}$ is reached}\label{line:finish}
		\Else
			\State Append $\links_{j_l}(v)$ to $A_{j_l}$\Comment{append predetermined vertex}\label{line:predetRepl}
			\State Mark $\links_{j_l}(v)$ as active
		\EndIf
		\State Unmark $v$ from being active
	\EndWhile
	\State Output $A_1,\ldots,A_k$
\end{algorithmic}
\end{algorithm}

However, if at least two paths end at $v$, this approach does not ensure vertex-disjointness. Let $A_{j_1},\ldots,A_{j_l}$ these $l \geq 2$ paths and assume $j_1 < j_2 < \cdots < j_l$. We first replace the end vertex $v$ of $A_{j_i}$ with the vertex $\links_{j_{i-1}}(\vorletzter(A_{j_i}))$ for all $i \neq 1$. We will show that these modified end vertices are strictly smaller than $v$, which will re-establish the vertex-disjointness. The key idea of the algorithm is then to switch the indices of the $l$ paths appropriately such that the appended vertices are again predetermined by the path index.

Let a \emph{cyclic downshift} on $A_{j_1},\ldots,A_{j_l}$ replace the index of each path by the next smaller index of a path in this set (where the next smaller index of $j_1$ is $j_l$), i.e.\ we set $A_{j_i} := A_{j_{i+1}}$ for every $i \neq l$ and then replace $A_{j_l}$ with the old path $A_{j_1}$. We perform a cyclic downshift on $A_{j_1},\ldots,A_{j_l}$. Note that we did not alter the path $A_{j_l}$ (which was named $A_{j_1}$ before) yet. If $v = r_{j_l}$, $A_{j_l}$ is finished; otherwise, we append the vertex $\links_{j_l}(v)$ to $A_{j_l}$. See Algorithm~\ref{alg:looseEnds} for a description of the algorithm in pseudo-code. Figure~\ref{bsp:durchlauf} shows a run of Algorithm~\ref{alg:looseEnds}.

\begin{figure}[p]
   \captionsetup[subfloat]{justification=justified,singlelinecheck=false}
   \centering
      \subfloat[A MAO of a graph $G$ and its forests $F_1$ (green), $F_2$ (red, dashed) and $F_3$ (blue, dotted).]
      {\makebox[13cm]{\includegraphics[scale=0.55,page=1]{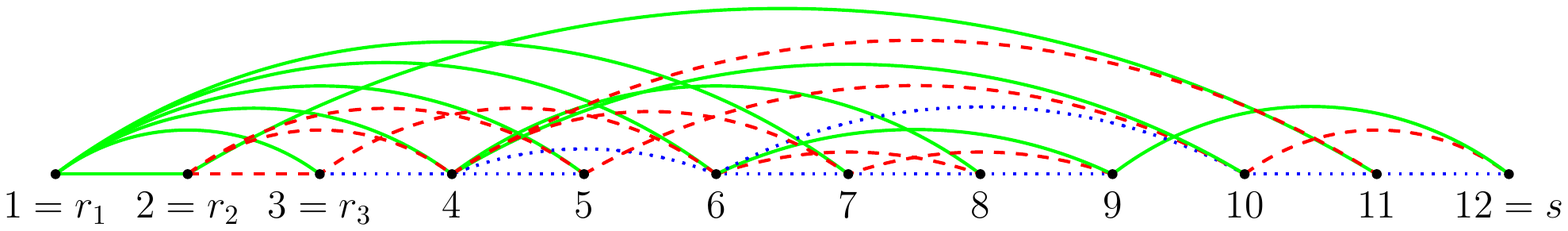}}}\\
      \subfloat[Paths $A_1$ (green), $A_2$ (red, dashed) and $A_3$ (blue, dotted) after the initialization phase and processing vertex $11$. The paths $A_2$ and $A_3$ end at the largest active vertex $10$.]
      {\makebox[13cm]{\includegraphics[scale=0.55,page=2]{beispieldurchlauf}}}\\
      \subfloat[After processing vertex $10$, the paths $A_2$ and $A_3$ have been shifted, which is here depicted by a color change. The last vertex of $A_2$ is then replaced, while $A_3$ is extended in $F_3$.]
      {\makebox[13cm]{\includegraphics[scale=0.55,page=3]{beispieldurchlauf}}}\\
      \subfloat[After processing $9$, the largest active vertex is $6$.]
      {\makebox[13cm]{\includegraphics[scale=0.55,page=4]{beispieldurchlauf}}}\\
      \subfloat[After shifting and extending $A_1$ and $A_3$, all three paths meet at the largest active vertex $4$.]
      {\makebox[13cm]{\includegraphics[scale=0.55,page=5]{beispieldurchlauf}}}\\
      \subfloat[Downshift: The old path $A_3$ is now $A_2$, the old $A_2$ is now $A_1$ and the old $A_1$ is now $A_3$.]
      {\makebox[13cm]{\includegraphics[scale=0.55,page=6]{beispieldurchlauf}}}\\
      \subfloat[After processing root $r_3=3$, $A_3$ is finished and $A_1$ and $A_2$ are shifted.]{\makebox[13cm]{\includegraphics[scale=0.55,page=7]{beispieldurchlauf}}}\\
      \subfloat[After processing the roots $r_2=2$ and $r_1=1$, the paths $A_2$ and $A_3$ are finished.]
      {\makebox[13cm]{\includegraphics[scale=0.55,page=8]{beispieldurchlauf}}}
   \caption{A run of Algorithm~\ref{alg:looseEnds} on the graph depicted in (a) when $s=12$ and $k=3$.}\label{bsp:durchlauf}
\end{figure}

We prove the correctness of Algorithm~\ref{alg:looseEnds}. Before the processing step of any active vertex $v$, the $A_i$ satisfy several invariants, the most crucial of which are that they are $\{v+1,\ldots,s-1\}$-disjoint and that the vertices of every $A_i$ are decreasing in $<$. In detail, we have the following invariants.

\begin{invariants}
Let $v < s$ be the largest active vertex, or $v := 0$ if there is no active vertex left. Before processing $v$, the following invariants are satisfied for every $1 \leq i \leq k$:
\begin{enumerate}
	\item[(1)] The vertices of $A_i$ start with $s$ and are strictly decreasing in $<$.
	\item[(2)] The path $A_i$ is finished if and only if $\letzter(A_i) > v$. In this case, $\letzter(A_i) = r_i$.\\
	If $A_i$ is not finished, $r_i\leq \letzter(A_i) \leq v$ and the last edge of $A_i$ is in $T_i$.
	\item[(3)] $\vorletzter(A_i) > v$ 
	\item[(4)] Every vertex $w \in A_i$ satisfying $v < w < s$ is not contained in any $A_j \neq A_i$.
	\item[(5)] $A_i \subseteq T_1 \cup \cdots \cup T_k$
\end{enumerate}
\end{invariants}

We first clarify the consequences. Invariant~(2) implies that the algorithm has finished all paths $A_i$ precisely after processing $r_1$, and that every $A_i$ ends at $r_i$. The Invariants~(1) and (3) are necessary to prove Invariant~(4), which in turn implies that the $A_i$ are $\{v+1,\ldots,s-1\}$-disjoint before processing an active vertex $v$. Hence, the final paths $A_i$ are $\overline{s}$-disjoint. With Invariant~(5) this gives the claim of Lemma~\ref{lem:looseends}.

It remains to prove Invariants~(1)--(5). Immediately after initializing the paths $A_1,\ldots,A_k$, the next active vertex is $\letzter(A_k) < s$. It is easy to see that all five invariants are satisfied for $v = \letzter(A_k)$, i.e.\ before processing the first active vertex. We will prove that processing any largest active vertex $v$ preserves all five invariants for the active vertex $v'$ that follows $v$ (where $v':=0$ if $v$ is only remaining active vertex). For this purpose, let $A'_i$ be the path with index $i$ immediately before processing $v'$ and let $A_i$ be the path with index $i$ before processing $v$; by hypothesis, the $A_i$ satisfy all invariants for $v$.

For Lines~\ref{line:Repl} and~\ref{line:predetRepl} in the processing step of $v$, we have to prove the existence of $\links_{j_{i-1}}(\vorletzter(A_{j_i}))$ and $\links_{j_l}(v)$ respectively. In Line~\ref{line:Repl}, we have $i \geq 2$ and $\letzter(A_{j_i}) = v$. Then Invariant~(2) implies that $A_{j_i}$ is not finished and $v = \letzter(A_{j_i}) = \links_{j_i}(\vorletzter(A_{j_i}))$. Thus, $\links_{j_{i-1}}(\vorletzter(A_{j_i}))$ exists. In Line~\ref{line:predetRepl}, we have $v\neq r_{j_l}$ and $\letzter(A_{j_l}) = v$ ($A_{j_l}$ before the cyclic downshift). Then Invariant~(2) implies that $r_{j_l} \leq v$. This proves $r_{j_l}<v$ and the existence of $\links_{j_l}(v)$.

We prove $v' < v$ next. Consider the vertices that are newly marked as active in the processing step of $v$. According to Line~\ref{line:enda} of Algorithm~\ref{alg:looseEnds}, every such vertex is the new end vertex of some path $A_{j_i}$ with end vertex $v$ that was modified in the processing step of $v$ (we do not count index transformations as modifications). There are exactly two cases how $A_{j_i}$ may have been modified, namely either by Line~\ref{line:Repl} (then $2 \leq i \leq l$ and $ \links_{j_{i-1}}(\vorletzter(A_{j_i}))$ is the vertex that is newly marked as active) or by Line~\ref{line:predetRepl} (then $\links_{j_l}(v)$ is the vertex that is newly marked as active); in particular, $A_{j_i}$ was not modified by both lines. In the first case, $A_{j_i}$ satisfies Invariant~(2) before the processing step of $v$ by hypothesis. In fact, we have $r_{j_i}\leq v$, as $v< r_{j_i}$ implies that $A_{j_i}$ is finished and since $\letzter(A_{j_i})>v$ would contradict $\letzter(A_{j_i}) = v$.

Hence, the last edge of $A_{j_i}$ is in $T_{j_i}$, which shows $v = \links_{j_i}(\vorletzter(A_{j_i}))$. Since $j_{i-1} < j_i$ by Line~\ref{line:enda} and due to Corollary~\ref{cor:leftedges}, we conclude $\links_{j_{i-1}}(\vorletzter(A_{j_i})) < v$. In the second case, Corollary~\ref{cor:leftedges} implies $\links_{j_l}(v) < v$. Thus, in both cases, every new active vertex is strictly smaller than $v$, which proves $v'<v$.

This gives Invariant~(1), as every $A'_{j_i}$ starts with $s$ and every new vertex is left of its predecessor in the path by Corollary~\ref{cor:leftedges}.

For Invariant~(2), consider the path $A'_i$ for any $i$. First, assume that $A'_i$ is finished. Then either $A_i$ is finished or $v=r_{i}$, according to Line~\ref{line:finish} of Algorithm~\ref{alg:looseEnds} in the processing step of $v$. In the former case, $A_i$ satisfies Invariant~(2) for $v$ and so does $A'_i$ for $v'<v$. In the latter case, we have $v'<v=r_{i}$ and $\letzter(A'_i) = \letzter(A_{j_1}) = v$.

Second, assume that $A'_i$ was not modified in the processing step of $v$ and is not finished. Then $\letzter(A'_i) < v$, as every path with end vertex at least $v$ is modified or finished in the processing step of $v$ or finished before. In particular, processing $v$ did not change the index of $A_i = A'_i$. As $A_i$ satisfies Invariant~(2) for $v$ by hypothesis, the only condition of Invariant~(2) that may be violated for $v'$ is $\letzter(A'_i) < v'$. However, as $\letzter(A'_i) < v$ was marked as active in some previous step of Algorithm~\ref{alg:looseEnds} and since $v'$ is the largest active vertex, $\letzter(A'_i) \leq v'$. Thus, $A'_i$ satisfies Invariant~(2) for $v'$.

Third, assume that $A'_{j_i}$ was modified in the processing step of $v$ and is not finished. Then $A'_{j_i}$ was modified either by Line~\ref{line:Repl} or~\ref{line:predetRepl}. If $A'_{j_i}$ was modified by Line~\ref{line:Repl}, we have $i < l > 1$ after the cyclic downshift, as the path $A_{j_1}$ is not modified by Line~\ref{line:Repl}. In addition, we know $\letzter(A'_{j_i}) = \links_{j_i}(\vorletzter(A_{j_{i+1}}))< \links_{j_{i+1}}(\vorletzter(A_{j_{i+1}})) =v$ by Corollary~\ref{cor:leftedges} and that the last edge of $A'_{j_i}$ is in $T_{j_i}$. Thus, $r_{j_i} \leq \letzter(A'_{j_i})$. If $A'_{j_i}$ was modified by Line~\ref{line:predetRepl}, we have $i=l$ and $r_{j_l} \leq \links_{j_l}(v) = \letzter(A'_{j_l})$ by Corollary~\ref{cor:leftedges}. Then the last edge of $A'_{j_l}$ is in $T_{j_l}$. In both cases, $\letzter(A'_{j_l})$ is active before processing $v'$ and it follows $\letzter(A'_{j_l}) \leq v'$.

For Invariant~(3), assume to the contrary that $\vorletzter(A'_i) \leq v'$. Since $v' < v < \vorletzter(A_j)$ for all $j\in \{1, \dots, k\}$, a new end vertex was appended to $A'_i$ in the processing step of $v$ (the end vertex was not replaced, as this would not have changed $\vorletzter(A'_i)$). This must have been done in Line~\ref{line:predetRepl} of Algorithm~\ref{alg:looseEnds} and we conclude $v' < v = \vorletzter(A'_i)$, which contradicts the assumption.

For Invariant~(4), consider Line~\ref{line:Repl} of the processing step of $v$. As showed in the proof of $v'<v$ above, we have $\links_{j_{i-1}}(\vorletzter(A_{j_i})) < v$ for all $1 < i \leq l$. Thus, Invariants~(1) and (3) imply that exactly the path $A'_{j_l}$ of the paths $A'_1,\ldots,A'_k$ contains $v$.

Invariant~(5) follows directly from the definition of $\links_i$. This concludes the proof of Lemma~\ref{lem:looseends}.

\begin{remark}
Invariant~(4) cannot be strengthened to $A_i \subseteq T_1 \cup \cdots \cup T_i$ or $A_i \subseteq T_i \cup \cdots \cup T_k$, as the cyclic downshifting may force $A_i$ to contain edges of $F_k$ or $F_1$: During the construction, many active vertices $v \in T_k$ may occur at which all $k$ paths (temporarily) end (see Figure~\ref{bsp:buntePfade}, for which each of the three paths $A_1, A_2, A_3$ contains edges of every tree $T_1, T_2, T_3$).

\begin{figure}[h!t]
   \captionsetup[subfloat]{justification=justified,singlelinecheck=false}
   \centering
      \subfloat[A MAO of a graph $H_1$ and its forests $F_1$ (green), $F_2$ (red, dashed) and $F_3$ (blue, dotted).]{\includegraphics[scale=0.65,page=1]{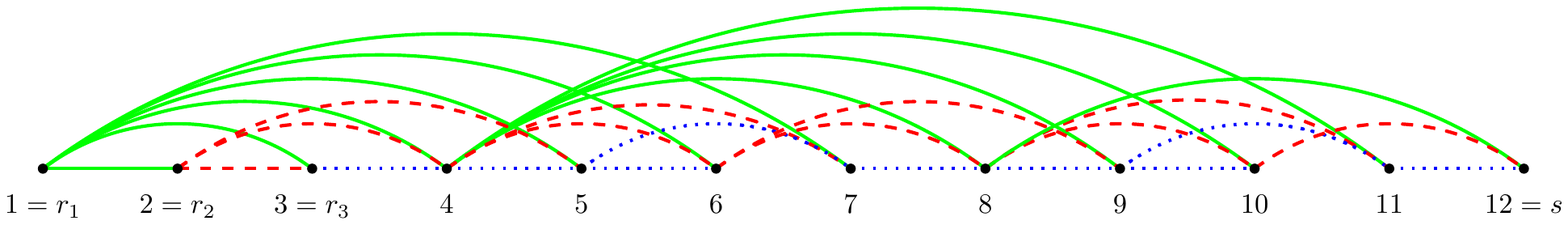}}\\
      \subfloat[The $\overline s$-disjoint paths $A_1$, $A_2$ and $A_3$ in $H_1$.]{\includegraphics[scale=0.65,page=2]{example01b-ipe}}
   \caption{The output of Algorithm~\ref{alg:looseEnds} on the graph depicted in (a) when $s=12$ and $k=3$.} \label{bsp:buntePfade}
\end{figure}
\end{remark}

So far we have shown an algorithmic proof for the existence of $k$ $\overline{s}$-disjoint paths from $s$ to the roots $r_1,\ldots,r_k$. It remains to show the running time for Lemma~\ref{lem:looseends}. Algorithm~\ref{alg:looseEnds} can be implemented in time $O(n+m)$, as it visits every edge only a constant number of times. Recall that the input graph $G$ does not have to be $k$-connected; it suffices that $G$ has minimum degree of at least $k$. A closer look at Algorithm~\ref{alg:looseEnds} shows that actually only the edges in $T_1 \cup \dots \cup T_k$ are visited, which gives the slightly better running time $O(|E(T_1 \cup \dots \cup T_k)|)$.

\section{Computing Vertex-Disjoint Paths Between Two Vertices}
We use the algorithm of the last section to prove our following main result.

\begin{theorem}\label{thm:main}
Let $t < s$ be a vertex in $T_k$. Then $k$ internally vertex-disjoint paths between $s$ and $t$ can be computed in time $O(|E(T_1 \cup \dots \cup T_k)|) \subseteq O(n+m)$.
\end{theorem}

This theorem is directly implied by the following lemma.

\begin{lemma}\label{lem:paths}
Let $t < s$ be a vertex in $T_k$. Then there are $k$ paths $A_1,\dots,A_k$ with start vertex $s$ and $k$ paths $B_1,\dots,B_k$ with start vertex $t$ such that $\letzter(A_i) = \letzter(B_i)$ for every $i$ and $\{A_1 \cup B_1,\dots,A_k \cup B_k\}$ is a set of $k$ internally vertex disjoint paths from $s$ to $t$. Moreover, all paths are contained in $T_1 \cup \cdots \cup T_k$ and can be computed in time $O(|E(T_1 \cup \dots \cup T_k)|)$.
\end{lemma}

A first idea would be to use the loose ends-algorithm twice, once for the start vertex $s$ and once for the start vertex $t$, in order to find the paths $A_i$ and $B_i$ for all $i$. However, in general this is bound to fail, as Figure~\ref{fig:matchingstattloose} shows. A second attempt may try to finish two paths $A_i$ and $B_j$ whenever they end at the same active vertex. However, this may fail when $i \neq j$, as then two single paths $A_{i'}$ and $B_{j'}$ may remain that end at the respective roots $r_{i'}$ and $r_{j'} > r_{i'}$ such that $B_{j'}$ cannot be extended to $r_{i'}$ without violating the index scheme of Invariant~(2), as Figure~\ref{fig:unpassendeenden} shows.

\begin{figure}[!ht]
   \captionsetup[subfloat]{justification=justified,singlelinecheck=false}
   \centering
      \subfloat[A MAO of a graph $H_2$ and its forests $F_1$ (green), $F_2$ (red, dashed) and $F_3$ (blue, dotted).]
      {\makebox[13cm]{
      \includegraphics[scale=0.75,page=1]{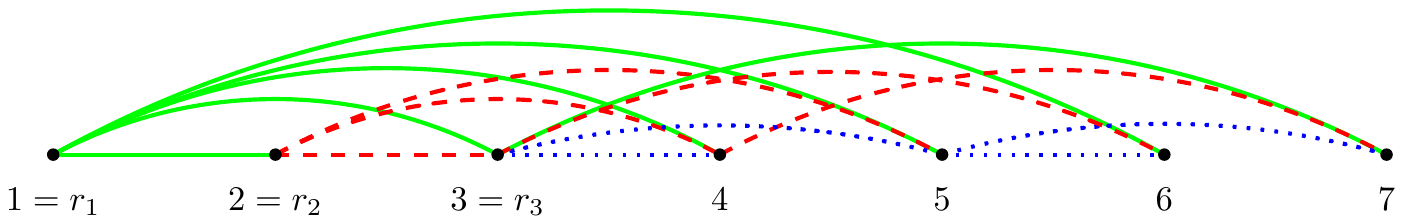}}}\\
      \subfloat[The output of Algorithm~\ref{alg:looseEnds} on $H_2$ with start vertex $s=7$ and $k=3$ (with a cyclic downshift at 3). Note that each final path $A_i$ ends at $r_i$ due to Invariant~(2). The final $A_1$ contains $5$.]
      {\makebox[13cm]{
      \includegraphics[scale=0.75,page=2]{example02b-ipe}}}\\
      \subfloat[The output of Algorithm~\ref{alg:looseEnds} on $H_2$ with start vertex $t=6$ and $k=3$. Since $B_2$ contains 5, both outputs contain the vertex $5 \notin \{s,t,r_1,r_2,r_3\}$, preventing the disjointess of their union.]
      {\makebox[13cm]{
      \includegraphics[scale=0.75,page=3]{example02b-ipe}}}\\
      \subfloat[The union of the two outputs. The vertices $s=7$ and $t=6$ are separated by $\{3,5\}$.]
      {\makebox[13cm]{
      \includegraphics[scale=0.75,page=4]{example02b-ipe}}}\\
   \caption{A graph $H_2$ for which applying Algorithm~\ref{alg:looseEnds} twice does not give internally vertex-disjoint paths.}\label{fig:matchingstattloose}
\end{figure}

\begin{figure}[!ht]
   \captionsetup[subfloat]{justification=justified,singlelinecheck=false}
   \centering
      \subfloat[A MAO of a graph $H_3$ and its forests $F_1$ (green), $F_2$ (red, dashed) and $F_3$ (blue, dotted).]
      {\makebox[13.5cm]{
      \includegraphics[scale=0.75,page=1]{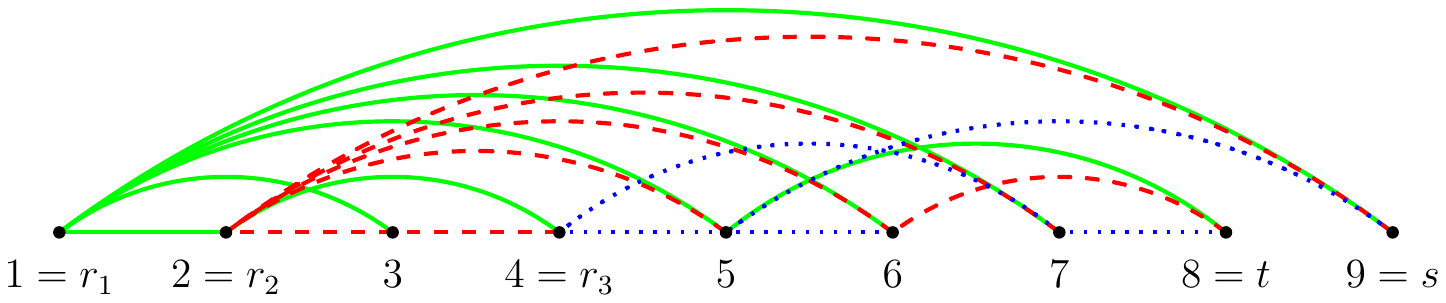}}}\\
      \subfloat[The output of Algorithm~\ref{alg:looseEnds} on $H_3$, $k=3$, with start vertices $s=9$ and $t=8$, respectively, such that the path $A_3$ from $s$ and the path $B_1$ from $t$ are finished at vertex 5. Note that the path $B_3$ ends at root $r_3=4$ without hitting a matching path from $s$.]
      {\makebox[13.5cm]{
      \includegraphics[scale=0.75,page=2]{example03-ipe}}}\\
   \caption{A graph $H_3$ for which applying Algorithm~\ref{alg:looseEnds} twice does not give internally vertex-disjoint paths if paths with the same end vertex are finished.}\label{fig:unpassendeenden}
\end{figure}

We will nevertheless use Algorithm~\ref{alg:looseEnds} to prove Lemma~\ref{lem:paths}, but in a more subtle way, as outlined next. First, we compute the paths $A_1,\dots,A_k$ with start vertex $s$ using Algorithm~\ref{alg:looseEnds}, until the largest active vertex $v$ is less or equal $t$ (i.e.\ the parts of the $A_i$ between $s$ and $t$ are just computed by Algorithm~\ref{alg:looseEnds}). As soon as $v\leq t$, we additionally construct a second set of paths $B_1,\dots,B_k$ with start vertex $t$ using Algorithm~\ref{alg:looseEnds}.

The main difference to Algorithm~\ref{alg:looseEnds} from this point on is that we extend the paths $A_i$ and the paths $B_i$ in parallel (i.e.\ we take the largest active vertex of both running constructions) such that, after the processing step of $v$, the vertex $v$ is not contained in any two paths $A_i$ and $B_j$ with $i \neq j$. This ensures the vertex-disjointness.

If no $A$-path or no $B$-path ends at $v$, we again just perform Algorithm~\ref{alg:looseEnds}; then at most one path contains $v$ after the processing step. Otherwise, some $A$-path and some $B$-path ends at $v$. After the processing step at $v$, we want to have exactly two paths $A_j$ and $B_j$ (i.e.\ having the same index) that end at $v$; such a pair of paths is then finished. In order to ensure this, we choose $j$ as the largest index such that $A_j$ \emph{or} $B_j$ ends at $v$ before processing $v$. If both $A_j$ and $B_j$ end at $v$, we perform one processing step of Algorithm~\ref{alg:looseEnds} at $v$ for the $A$-paths and the $B$-paths, respectively, which implies that no other path is ending at $v$.

Otherwise, exactly one of the paths $A_j$ and $B_j$ ends at $v$, say $A_j$. Then $B_j$ is not finished, as we finish only paths having the same index, and the last edge of $B_j$ is in $F_j$. Hence, there is an index $i<j$ such that $B_i$ ends at $v$. We then apply a processing step of Algorithm~\ref{alg:looseEnds} (including a cyclic downshift) on $B_j$ and all $B$-paths that end at $v$, and one on all $A$-paths, respectively. Then the new paths $A_j$ and $B_j$ (due to cyclic downshifts, these correspond to the former $A$- and $B$-paths with lowest index ending at $v$) end at $v$ afterward, but no other $A$- or $B$-path, as desired. Note that the replacement of the last edge of (the old) $B_j$, which did not end at $v$ but, say, at a vertex $w$, may cause $w$ to be active although neither an $A$-path nor a $B$-path ends at $w$.

\begin{algorithm}[p]
\caption{MatchingEnds($G,<,s,t,k$)\Comment{$t$ is a vertex in $T_k$, $t<s$}}\label{alg:matchingEnds}
\begin{algorithmic}[1]
	\ForAll{$i$}\Comment{initialize all $A_i$ and $B_i$}
		\State $A_i := (s,\links_i(s))$
		\State Mark $\links_i(s)$ as active
		\State $B_i := (t)$
	\EndFor
	\State Mark $t$ as active
	\While{there is a largest active vertex $v$}\Comment{process $v$} \label{line:while}
	\If {v=t}
		\ForAll{$i$}\Comment{initialize all $A_i$}
			\If {$\letzter(A_i)= t$}
				\State $A_{i}$, $B_i$ are finished \label{line:finish1}
			\Else
				\State Append $\links_i(t)$ to $B_i$ \label{line:predetB}
				\State Mark $\links_i(t)$ as active \label{line:active1}
			\EndIf
		\EndFor
		\State Unmark $t$ from being active
	\Else
		\State $I_A:= \{i \vert \letzter(A_i) = v\}$\label{line:end-A}
		\State $I_B:= \{i \vert \letzter(B_i) = v\}$\label{line:end-B}
		\If {$I_A$ and $I_B$ are empty}
			\State Unmark $v$ from being active and go to Line \ref{line:while} \label{line:cancel}
		\EndIf
		\State $j := \max(I_A\cup I_B)$ \label{line:def-j}
		\ForAll{pairs $(i_1, i_2)$ of consecutive	indices $i_1<i_2$ in $I_A \cup \{j\}$}\Comment{replace ends}
			\State Replace $\letzter(A_{i_2})$ with $\links_{i_1}(\vorletzter(A_{i_2}))$\label{line:ReplA}
			\State Mark $\links_{i_1}(\vorletzter(A_{i_2}))$ as active
		\EndFor
		\ForAll{pairs $(i_1, i_2)$ of consecutive	indices $i_1<i_2$ in $I_B \cup \{j\}$}\Comment{replace ends}
			\State Replace $\letzter(B_{i_2})$ with $\links_{i_1}(\vorletzter(B_{i_2}))$\label{line:ReplB}
			\State Mark $\links_{i_1}(\vorletzter(B_{i_2}))$ as active
		\EndFor
		\State Perform a cyclic downshift on all $A_i$ with $i\in I_A \cup {j}$\label{line:DownshiftA}
		\State Perform a cyclic downshift on all $B_i$ with $i\in I_B \cup {j}$\label{line:DownshiftB}
		\If{$v = \letzter(A_j) = \letzter (B_j)$} \Comment{if and only if $I_A \neq \emptyset \neq I_B$}\label{line:bed-finish2}
			\State $A_{j}$, $B_j$ are finished \label{line:finish2}
		\ElsIf {$v=\letzter(A_j)$}
			\State Append $\links_{j}(v)$ to $A_{j}$\Comment{append predetermined vertex}\label{line:predetReplA}
			\State Mark $\links_{j}(v)$ as active
		\ElsIf {$v=\letzter(B_j)$}
			\State Append $\links_{j}(v)$ to $B_{j}$\Comment{append predetermined vertex}\label{line:predetReplB}
			\State Mark $\links_{j}(v)$ as active
		\EndIf
		\State Unmark $v$ from being active
	\EndIf
	\EndWhile
	\State Output $A_1, \dots, A_k$, $B_1, \dots, B_k$
\end{algorithmic}
\end{algorithm}

For a precise description of the approach, see Algorithm~\ref{alg:matchingEnds}. We now show that Algorithm~\ref{alg:matchingEnds} outputs the desired paths and thus proves Lemma~\ref{lem:paths}. The following observations follow directly from Algorithm~\ref{alg:matchingEnds}.

\begin{observation}\label{obs:1}
Throughout Algorithm~\ref{alg:matchingEnds} the paths $A_1,\dots,A_k,B_1,\dots,B_k$ satisfy the following properties.
\begin{itemize}
	\item[(1)] For every $i\in \{1, \dots, k\}$, $A_i$ and $B_i$ are both finished or both unfinished.
	\item[(2)] As long as the largest active vertex is larger than $t$, $B_1 = B_2 = \dots = B_k = (t)$.
	\item[(3)] The end vertex of every unfinished path is active.
\end{itemize}
\end{observation}

We prove Lemma~\ref{lem:paths} by showing that Algorithm~\ref{alg:matchingEnds} outputs the desired paths. Before the processing step of any active vertex $v$, the paths $A_i$ and $B_i$ satisfy several invariants, the most crucial of which are that they are $\{v+1,\ldots,s-1\}\backslash\{t\}$-disjoint and that the vertices of every $A_i$ and $B_i$ are decreasing in $<$. We will prove the following invariants.

\begin{invariants}
Let $v < s$ be the largest active vertex, or $v := 0$ if there is no active vertex left. Before processing $v$, the following invariants are satisfied for every $1 \leq i \leq k$:
\begin{enumerate}
	\item[(1)] $A_i$ starts with $s$, $B_i$ starts with $t$, and the vertices of both paths are strictly decreasing in $<$.
	\item[(2)] The paths $A_i$ and $B_i$ are finished if and only if $v<\letzter(A_i) = \letzter(B_i)$. If $A_i$ and $B_i$ are not finished, then $r_i\leq \letzter(A_i) \leq v$, $r_i \leq \letzter(B_i) \leq v$, and the last edge of $A_i$ as well as the last edge of $B_i$ (if $B_i$ has length at least $1$) are in $T_i$.
	\item[(3)] $\vorletzter(A_i) > v$. If $v\geq t$, $B_i = (t)$. If $v<t$, either $B_i$ is finished with $B_i = (t)$ or $B_i$ has length at least $1$ such that $\vorletzter(B_i)>v$.  
	\item[(4)] Let $w \neq t$ be a vertex with $v<w<s$. If $w \in A_i \cup B_i$, $w$ is neither contained in a path $A_j \neq A_i$ nor in a path $B_j \neq B_i$. If $w \in A_i \cap B_i$, $A_i$ and $B_i$ are finished with $w = \letzter(A_i)= \letzter(B_i)$.
	\item[(5)] $A_i \cup B_i \subseteq T_1 \cup \cdots \cup T_k$
\end{enumerate}
\end{invariants}

Before proving these invariants, we first clarify some of their consequences. Invariant~(2) implies that the algorithm has finished all paths when $v = 0$ and that the end vertices of $A_i$  and $B_i$ match for all $i$. Invariants~(1) and (3) will be necessary to prove Invariant~(4), which in turn implies that the paths $A_1 \cup B_1,\dots,A_k \cup B_k$ are internally vertex-disjoint. Invariant~(5) settles the first part of the second claim of Lemma~\ref{lem:paths}. We continue with further consequences of some of these invariants, which will be used later.

\begin{observation}\label{obs:2}
Let $v < s$ be the largest active vertex, or $v := 0$ if there is no active vertex left. Before processing $v$, we have the following observations:
\begin{itemize}
	\item[(1)] Assume Invariants~(1) and (3). Then, for every $1 \leq i \leq k$, all vertices of the paths $A_i$ and $B_i$ except $\letzter(A_i)$ and $\letzter(B_i)$ are greater than $v$ before processing $v$.
	\item[(2)] Assume Invariant~(2). Then no finished path is modified while processing $v$, as Algorithm~\ref{alg:matchingEnds} modifies $A_i$ or $B_i$, $1 \leq i \leq k$, only if at least one of them ends at $v$.
	\item[(3)] Assume Invariants~(2) and~(3). Then the largest active vertex after processing $v > 0$ is smaller than $v$.
	
	The following proof uses the fact that new active vertices are only created after Lines~\ref{line:ReplA} and~\ref{line:ReplB}. When Line~\ref{line:ReplA} replaces the end vertex of $A_{i_2}$, the new end vertex $\links_{i_1}(\vorletzter(A_{i_2}))$ is smaller than the old one due to Corollary~\ref{cor:leftedges}. By the previous observation, $A_{i_2}$ is not finished. According to Invariant~(2), the new end vertex of $A_{i_2}$ is smaller than $v$.
	
	When invoking Line~\ref{line:ReplB}, $|I_B| \geq 2$ and, hence, there is at least one path $B_i$ that ends at $v \neq t$: According to Invariant~(3), $v < t$ and $B_{i_2}$ has length at least one such that $\vorletzter(B_i)>v$. Thus, the new end vertex of $B_{i_2}$ is smaller than the old one and, in particular, smaller than $v$.
\end{itemize}
\end{observation}

We now prove Invariants~(1)--(5) for all largest active vertices encountered in the algorithm. Immediately after initializing the paths $A_1,\ldots,A_k$ and $B_1, \ldots, B_k$, the largest active vertex is $\letzter(A_k) < s$ or $t$. It is easy to see that all five invariants are satisfied for $v = \max\{ \letzter(A_k),t\}$, i.e.\ before processing the first active vertex. We will prove that processing any largest active vertex $v$ preserves all five invariants for the largest active vertex $v'$ after having processed $v$ (here, we set $v':=0$ if $v$ is the last active vertex). For this purpose, let $A_i$ and $B_i$ be the paths with index $i$ immediately before processing $v$ and let $A'_i$ and $B'_i$ be the paths with index $i$ immediately before processing $v'$. By induction hypothesis, $A_i$ and $B_i$ satisfy the desired invariants for the vertex $v$; hence, we can also use their implications in Observation~\ref{obs:2}. In particular, $v' < v$ due to Observation~\ref{obs:2}.(3).

We assume first $v \neq t$ and that $v$ is not the end vertex of any of the paths $A_1,\dots,A_k,B_1,\dots,B_k$ (the following cases can thus assume that at least one such path ends at $v$). Then Line~\ref{line:cancel} cancels the processing step without changing any path. According to Observation~\ref{obs:2}.(1), $v$ is not contained in any of the paths $A_1, \dots, A_k, B_1, \dots, B_k$. The paths $A'_1=A_1, \dots, A'_k=A_k$, $B'_1=B_1, \dots, B'_k=B_k$ clearly satisfy Invariants~(1) and~(5) also for $v'$. They also satisfy Invariant~(3) for $v'$, since $v'\geq t$ if $v>t$, and Invariant~(2) for $v'$, since the largest end vertex of any unfinished path is at most $v'$, according to Observation~\ref{obs:1}.(3). By Invariant (2) for $v$ and Observation~\ref{obs:2}.(1), any vertex $w \neq t$ with $v'<w\leq v$ is contained in $A_i$ or $B_i$ only if $A_i$ or $B_i$ is unfinished and has end vertex $w$. Then $w$ is active due to Observation~\ref{obs:1}.(3), which contradicts the choice of $v'$. Any $w>v$ would give Invariant~(4) by induction hypothesis. Thus, Invariant (4) is satisfied for $v'$.

Next, we assume $v>t$. Then the paths $A_1,\dots,A_k$ are all not finished and, according to Invariant~(3), $B_1 = \dots = B_k = (t)$. We thus have $I_B=\emptyset$ in Line~\ref{line:end-B} and $j=\max I_A$ in Line~\ref{line:def-j}. In particular, the paths $B_1, \dots, B_k$ are not modified while processing $v$ and the changes of $A_1, \dots, A_k$ are identical to the ones in the loose ends algorithm. Thus, $A'_1, \dots, A'_k$ satisfy the invariants of the loose ends algorithm for $v'$, which imply the invariants of Algorithm~\ref{alg:matchingEnds} for $v'$, as $t \leq v'$.

We assume $v=t$. Then $\links_i(t)$ exists in Line~\ref{line:predetB}, since we have $\letzter(A_i) \neq t$, which implies by Invariant~(2) that $A_i$ is not finished and $r_i \leq \letzter(A_i) < t=v$. The paths $A'_1,\dots,A'_k,B'_1,\dots,B'_k$ satisfy Invariant~(1) for $v'$ due to $\links_i(t)<t$ for every $i$, Observation~\ref{obs:1}.(2) and Invariant~(1) for $v$. We prove Invariant~(2) for $v'$: If $A'_i$ and $B'_i$ are finished, $\letzter(A_i)=t$ before processing $v$ and we conclude $v' < \letzter(A'_i)=\letzter(B'_i) = t$. Otherwise, the algorithm leaves both $A'_i$ and $B'_i$ unfinished by setting $A'_i := A_i$ and $B'_i := (t, \links_i(t))$. Since both $\letzter(B'_i) = \links_i(t)$ and $\letzter(A'_i)=\letzter(A_i)$ are active after processing $t$, $r_i \leq \letzter(B'_i) \leq v'$ and $r_i \leq \letzter(A'_i) \leq v'$. In particular, both $A'_i$ and $B'_i$ are of length at least one and have their last edge in $T_i$, which gives Invariant~(2) for $v'$.

For Invariants~(3) and (5), note that $A'_i=A_i$ and $B'_i=(t)$ if $B'_i$ is finished and otherwise $A'_i=A_i$ and $B'_i=(t,\links_i(t))$. It remains to prove Invariant~(4). If $A_i$ contains a vertex $w$ with $v'<w<t$ (again, any larger $w$ would give Invariant~(4) by induction hypothesis), $w$ is the active end vertex of $A_i$ by Observations~\ref{obs:2}.(1) and~\ref{obs:1}.(3), which contradicts the choice of $v'$. Since $B'_i$ consists of $t$ and at most one active vertex that is smaller or equal to $v'$, the vertex $w$ of Invariant~(4) does not exist. This proves Invariant~(4) for $v'$.

It only remains to assume $v<t$. According to the statement of our first case, we can additionally assume $I_A \cup I_B \neq \emptyset$ (as defined in Lines~\ref{line:end-A} and~\ref{line:end-B}). Let $j$ be the index chosen in Line~\ref{line:def-j} of processing step $v$. Then Observation~\ref{obs:1}.(1) ensures that both $A_j$ and $B_j$ are unfinished; hence, the downshifts in Lines~\ref{line:DownshiftA} and~\ref{line:DownshiftB} are well-defined. The paths $A'_1,\dots,A'_k,B'_1,\dots,B'_k$ satisfy Invariant~(1) for $v'$ due to Invariant~(1) for $v$, the fact that $\links_i(v)<v$ for all $i$, and Observation~\ref{obs:2}.(3).

We prove Invariant~(2) for $v'$: By induction, this invariant is true for the paths that were finished before processing $v$. For all paths $A'_i$ and $B'_i$ that are unfinished after processing $v$, Observation~\ref{obs:1}.(3) ensures that $\letzter(A'_i) \leq v'$ and $\letzter(B'_i) \leq v'$.  The cyclic downshifts in processing step $v$ imply that the last edge of $A'_i$ and the last edge of $B'_i$ are in $T_i$ for every $i \neq j$. Since $A'_j$ and $B'_j$ are the only paths that may change their status from unfinished to finished during the processing step of $v$, this gives Invariant~(2) for all paths except for $A'_j$ and $B'_j$. These two paths are finished after processing step $v$ if and only if $v'< v = \letzter(A'_j) = \letzter(B'_j)$. In this case, Invariant~(2) for $v'$ is satisfied. The case $\letzter(A'_j) \neq v = \letzter(B'_j)$ (and, by symmetry, the case $\letzter(A'_j) = v \neq \letzter(B'_j)$) only occurs if $I_A = \emptyset$ in processing step $v$. In this case $A'_j=A_j$ and $\letzter(B'_j) = \links_j(v)$. By Observation~\ref{obs:1}.(3) and Invariant~(2) for $v$, $r_j\leq \letzter(A'_j) \leq v'$ and the last edge of $A_j$ is in $T_j$. This gives Invariant~(2) for $v'$.

In order to prove Invariant~(3) for $v'$, observe that all modified paths are unfinished and have length at least one. Thus, we only have to prove that, when appending a vertex to a path in processing step $v$, the new second last vertex (i.e.\ the old end vertex) is greater than $v'$. The algorithm does this only in Lines~\ref{line:predetReplA} and~\ref{line:predetReplB}, in which $v>v'$ is the new second last vertex.

We prove Invariant~(4) for $v'$. If $v<w<s$, this follows directly from Invariant~(4) for $v$, so let $v' < w \leq v$. First, suppose $w<v$ such that $w$ is contained in $A'_i \cup B'_i$. Then $w$ is active after processing $v$ by Invariants~(1) and~(3) for $v'$, which contradicts the choice of $v'$. Second, suppose $w=v$. Let $j$ be the index chosen in Line~\ref{line:def-j} of processing step $v$. If both paths $A'_j$ and $B'_j$ contain $v$, both end at $v$ and are finished by Line~\ref{line:finish2}. It remains to prove that $v=w$ is not contained in any other path than $A'_j$ and $B'_j$. If any path $A_i$ or $B_i$ contains $v$ before processing step $v$, it contains $v$ as end vertex by Observation~\ref{obs:2}(1). All paths with $v$ as end vertex, except for the one with smallest index, get a new end vertex in processing step $v$. After the cyclic downshift in processing step $v$, only the paths with index $j$ contain $v$. This proves Invariant~(4) for $v'$.

Invariant~(5) follows straight from the definition of $\links_i$. This concludes the proof of Invariants~(1)--(5) for every $v < s$.

As in the loose ends algorithm, the running time of Algorithm~\ref{alg:matchingEnds} is upper bounded by $O(|E(T_1 \cup \dots \cup T_k)|)$ and thus by $O(n+m)$, as it suffices to visit every edge in these trees $T_1,\dots,T_k$ only a constant number of times.

\subsection{Variants}
Several variants of Menger's theorem~\cite{Menger1927} are known. Instead of computing $k$ paths between two vertices, we can compute paths between a vertex and a set of vertices (\emph{fan variant}) and between two sets of vertices (\emph{set variant}). Our algorithm extends to these variants.

\begin{theorem}
Let $G$ be a simple graph and $<$, $s$ and $T_1,\dots,T_k$ be defined as in Section~\ref{sec:prel}.
\begin{itemize}
	\item[(i)] (Fan variant) Let $T=\{t_1,\dots,t_k\}$ be a subset of $V$ such that $r_i \leq t_i < s$ for every $i$. Then $k$ internally vertex-disjoint paths between $s$ and $T$ can be computed in time $O(|E(T_1 \cup \dots \cup T_k)|) \subseteq O(n+m)$.
	\item[(ii)] (Set variant) Let $T = \{t_1,\dots,t_k\}$ and $S=\{s_1, \dots, s_k\}$ be disjoint vertex sets such that $r_i \leq t_i <s$ and $r_i \leq s_i \leq s$ for every $i$. Then $k$ internally vertex-disjoint paths between $S$ and $T$ can be computed in time $O(|E(T_1 \cup \dots \cup T_k)|) \subseteq O(n+m)$.
\end{itemize}
\end{theorem}
\begin{proof}
Let $G'$ be the subgraph of $G$ that is induced by the vertex set $\{1,\dots,s\}$. Clearly, $<$ is also a MAO of $G'$ (restricted to the vertices $1,\dots,s$) and the relevant parts of the trees $T_1,\dots,T_k$ are preserved by Lemma~\ref{lem:interval}. For (i), augment $G'$ by a new vertex $s+1$ with $k$ edges to $t_1,\dots,t_k$; this preserves that $<$ is a MAO. Applying Theorem~\ref{thm:main} to this graph on the vertices $s$ and $s+1$ gives the claim. For (ii), we augment $G'$ by two new vertices of degree $k$ with neighborhoods $S$ and $T$, respectively, and apply Theorem~\ref{thm:main} on the vertices $s+1$ and $s+2$. In both cases, the running time is not increased.
\end{proof}

\paragraph{Mixed Connectivity.}
Let $G=(V,E)$ be a multigraph and let $\alpha: V \rightarrow \mathbb{N}^+$ be a weight function on its vertices. A set of paths connecting two vertices $s$ and $t$ of $G$ is called $\alpha$-\emph{independent} if every vertex $v \notin \{s,t\}$ is contained in at most $\alpha(v)$ of these paths. Let a multigraph be $\alpha$-\emph{simple} if the number of edges between every two vertices $v$ and $w$ is at most $\min\{\alpha(v),\alpha(w)\}$.

For $\alpha$-simple multigraphs $G$, Nagamochi~\cite{Nagamochi2006} generalized the existential variant of Theorem~\ref{thm:main} by showing that there are $k$ $\alpha$-independent $s$-$t$-paths, where $s$ and $t$ are chosen as in Theorem~\ref{thm:main}. It is possible to modify Algorithm~\ref{alg:matchingEnds} to compute also these paths without increasing its running time, by replacing the two cyclic downshifts by a more complicated algorithm to transform the path indices.

\paragraph{Acknowledgments.} We wish to thank Solomon Lo for pointing out a connection between MAOs and Mader's proof about pendant pairs.

\bibliographystyle{abbrv}
\bibliography{Jens}


\end{document}